\documentclass[runningheads]{llncs}
\usepackage[
left=2.5cm, right=2.5cm, top=3cm, bottom=3cm, 
headheight = 0.6cm,
footskip = 1cm
]{geometry}
\usepackage{amsmath,amssymb,amsfonts}
\usepackage{algorithmic}
\usepackage{algorithm}
\usepackage{cite} 
\usepackage{hyperref} 

\usepackage[T1]{fontenc}
\usepackage{graphicx}
\usepackage{color}

\begin{document}
\title{Max-Min Secrecy Rate Optimization for Secure ISAC Networks: Global Optimization and Low-Complexity Algorithm}
%
\titlerunning{Max-Min Secrecy Rate Optimization for Secure ISAC Networks}
%

\author{
Thanh-Nha To \and
Trung Quang Pham \and
Dang Y Hoang \and
Hoang-Lai Pham \and
Tuan Anh Pham
}
\authorrunning{T.N. To et al.}
%
\institute{
Viettel High Technology Industries Corporation, Viettel Group, Hanoi, Vietnam\\
\email{\{nhatt30, trungpq12, yhd10, haidv29, laiph3, tuanpa44\}@viettel.com.vn}
}
\maketitle              
\begin{abstract}
In this paper, we investigate a secure integrated sensing and communication (ISAC) system in which multiple communication users (CUs) coexist with multiple untrusted sensing users (SUs) that may eavesdrop on the confidential information intended for the CUs. To promote security fairness among users, we formulate a max–min secrecy rate optimization problem subject to a transmit power budget and sensing quality requirements characterized by beampattern matching error constraints. The resulting design problem is highly non-convex due to the secrecy rate expressions and non-convex sensing constraints. To address these challenges, we first reformulate the problem using semidefinite relaxation (SDR). Based on the reformulated problem, we develop a branch-and-bound (BB) framework combined with convex relaxations to obtain the globally optimal solution within a prescribed accuracy. To further reduce computational complexity, we propose a low-complexity algorithm based on successive convex approximation (SCA), which iteratively solves a sequence of convex subproblems and converges to a local solution. Numerical results demonstrate that the proposed BB algorithm achieves the global optimum and provides a benchmark for performance evaluation. Moreover, the proposed SCA-based algorithm attains near-optimal secrecy performance with significantly lower computational complexity, making it attractive for practical ISAC deployments. 

\end{abstract}

\section{Introduction}
\subsection{Overview}
With the rapid growth in both the number and functionality of smart devices, next-generation mobile networks are expected not only to support conventional communication services but also to simultaneously meet stringent requirements such as high data rates, low latency, strong security, and environmental sensing capabilities. In this context, integrated sensing and communication (ISAC) has emerged as a promising paradigm, enabling the integration of sensing and communication functionalities within a unified system. This integration significantly expands the range of applications, including unmanned aerial vehicles (UAVs), intelligent transportation systems, industrial automation, and next-generation mobile networks \cite{Liu25Toward, Liu22Survey}.

However, such multi-functional integration also introduces significant security challenges. Physical layer security (PLS) techniques have been widely recognized as an effective approach to ensuring system security by exploiting channel characteristics to achieve a higher transmission rate for legitimate users than that of potential eavesdroppers \cite{Wu22Secure}. Meanwhile, ISAC systems inherently involve trade-offs among multiple design objectives, particularly among secrecy rate, communication performance, and sensing quality. Optimizing one objective often leads to the degradation of others, making the system design problem highly complex and non-convex.

\subsection{Related Works} 
Recent studies have partially addressed different aspects of this problem. For instance, the use of active intelligent reflecting surfaces (Active IRS) has been proposed to enhance secrecy performance, while artificial noise (AN) techniques have been employed to degrade the eavesdropping capability of adversaries while maintaining sensing quality \cite{Bao24Secrecy}. In addition, globally optimal solutions have been explored to improve system performance in specific scenarios \cite{chen2025global, lu19Global}. By leveraging semidefinite relaxation (SDR), these challenging problems can be addressed by obtaining globally optimal solutions through rigorous proofs of relaxation tightness \cite{ren2022optimal}. Achieving such globally optimal solutions plays a crucial role in maximizing radar sensing performance or designing optimal transmit beamforming, while fully satisfying the signal-to-interference-plus-noise ratio (SINR) requirements for communication users under system constraints \cite{ren2023robust}.

In practice, the exact location of a sensing target is often unknown, which may cause the target to act as a potential eavesdropper and threaten communication security \cite{hou2024optimal}. To address this uncertainty, the worst-case secrecy rate criterion is commonly adopted to ensure robustness by maximizing the secrecy rate under the most unfavorable eavesdropping conditions. These criteria, combined with sensing quality constraints, result in a highly complex and non-convex optimization problem \cite{ren2024secure}. To effectively tackle this challenge, global optimization methods are often employed as a promising solution, enabling the design of globally optimal beamforming strategies that achieve a balance between communication security and sensing performance \cite{xu2023sensing, xu2025sensing}.

Nevertheless, despite the successful application of PLS techniques and global optimization approaches in ISAC systems, existing studies mainly focus on solving individual problems in isolation and are typically limited to single-antenna and single-user scenarios. Extending the global optimization framework to multi-user ISAC systems while ensuring secure communication introduces significantly higher mathematical complexity due to the coupled and high-dimensional nature of the problem. This constitutes a critical research gap that this work aims to address.

\subsection{Contributions} 
Nevertheless, existing works mainly rely on local optimization methods and rarely address fairness in secrecy performance in the presence of untrusted SUs. Motivated by this gap, we study a secure multi-user ISAC system where SUs may also act as potential eavesdroppers. The main contributions of this paper are summarized as follows:

\paragraph{System Modeling} We develop a multi-user ISAC framework in which multiple untrusted SUs simultaneously perform sensing and eavesdropping. We formulate a max--min secrecy rate optimization problem to ensure fairness among users, under practical transmit power and sensing quality constraints characterized by beampattern matching error.

\paragraph{Global Methodology} 
We propose a branch-and-bound (BB) method with convex relaxation, ensuring convergence to a global optimum. Tight upper and lower bounds are derived to improve computational efficiency and provide a benchmark for evaluating suboptimal methods.

\paragraph{Low-Complexity Methodology} 
We develop a low-complexity algorithm based on successive convex approximation (SCA) in order to overcome the computational complexity of overly large-scale problems in the BB framework. This practical approach iterates through a series of convex subproblems to efficiently find a high-quality local solution.

\paragraph{Simulations} Numerical results demonstrate that the proposed globally optimal design significantly outperforms conventional local optimization approaches, and reveal the performance gap between global and local solutions, especially in high-dimensional settings.

\section{System Model and Problem Formulation}
\subsection{Signal Model}
We consider a multi-user ISAC system, where a base station (BS) equipped with a uniform linear array (ULA) of $N_B$ antennas simultaneously serves $K$ single-antenna legitimate users and senses $J$ untrusted targets. These targets may also act as potential eavesdroppers capable of intercepting the information signals intended for legitimate users (e.g., UAV-like targets).
Let $\mathcal{K} = \{1,2,\dots,K\}$ and $\mathcal{J} = \{1,2,\dots,J\}$ denote the sets of communication users and sensing targets, respectively. The BS transmits information symbols $x_k \in \mathbb{C}$ to user $k$ and an artificial noise (AN) vector $\mathbf{x}_0 \in \mathbb{C}^{N_B}$ to enhance both sensing performance and communication security, where $\mathbb{E}\{|x_k|^2\} = 1$ for all $k \in \mathcal{K}$, $\mathbb{E}\{\mathbf{x}_0\} = \mathbf{0}$, and $\mathbb{E}\{\mathbf{x}_0 \mathbf{x}_0^H \} = \mathbf{I}_{N_B}$.

The transmitted signal at the BS is given by
\begin{align} 
\mathbf{s} = \mathbf{W}_0 \mathbf{x}_0 + \sum_{k \in \mathcal{K}} \mathbf{w}_k x_k,
\end{align}
where $\mathbf{W}_0 \in \mathbb{C}^{N_B \times N_B}$ denotes the AN beamforming matrix, and $\mathbf{w}_k \in \mathbb{C}^{N_B}$ is the beamforming vector for user $k$. For notational convenience, we define $\mathbf{W} = [\mathbf{W}_0, \mathbf{w}_1, \dots, \mathbf{w}_K] \in \mathbb{C}^{N_B \times (N_B + K)}$ as the collection of all transmit beamformers.

The total transmit power at the BS is expressed as
\begin{align}
P_{\mathrm{BS}}(\mathbf{W})= \mathbb{E} \{ \|\mathbf{s} \|^2 \} = \|\mathbf{W}_0\|_F^2 + \sum_{k \in \mathcal{K}} \|\mathbf{w}_k\|^2.
\end{align}
\subsection{Communication Model}

The achievable rate at legitimate user $k \in \mathcal{K}$ is given by
\begin{align}
R_k(\mathbf{W}) = \log_2\left(1 + \frac{|\mathbf{h}_{c,k}^H \mathbf{w}_k|^2}
{\| \mathbf{W}_0^H \mathbf{h}_{c,k}\|^2_F
+ \sum_{i \in \mathcal{K} \setminus \{k\}} |\mathbf{h}_{c,k}^H \mathbf{w}_i|^2 
+ \sigma_{c,k}^2} \right),
\end{align}
Similarly, the achievable eavesdropping rate of target $j \in \mathcal{J}$ when intercepting the signal of user $k \in \mathcal{K}$ is
\begin{align}
R_{j,k}(\mathbf{W}) = \log_2\left(1 + \frac{|\mathbf{h}_{s,j}^H \mathbf{w}_k|^2}
{\|\mathbf{W}_0^H \mathbf{h}_{s,j}\|^2_F
+ \sum_{i \in \mathcal{K} \setminus \{k\}} |\mathbf{h}_{s,j}^H \mathbf{w}_i|^2 
+ \sigma_{s,j}^2} \right).
\end{align}
The secrecy rate of user $k$, assuming the worst-case eavesdropping target, is defined as
\begin{align}
R_k^{\mathrm{sec}}(\mathbf{W}) =
\left[
R_k(\mathbf{W}) - \max_{j \in \mathcal{J}} R_{j,k}(\mathbf{W})
\right]^+,
\end{align}
where $[z]^+ = \max(z,0)$.

\subsection{Sensing Model}
Next, we consider target sensing, where the transmit beampattern matching error is used as the performance metric. In addition to communication, the BS performs sensing by transmitting probing signals and shaping the spatial beampattern toward directions of interest. 

We consider the line-of-sight (LoS) channel from the BS to untrusted targets, as commonly assumed in prior works \cite{ren2023robust, hou2024optimal}. Let $\theta_j$ denote the angle of departure (AoD) from the BS to target $j \in \mathcal{J}$. The steering vector corresponding to angle $\theta$ is given by $\boldsymbol{\nu}(\theta) = 
\left[1, e^{j\frac{2\pi d}{\lambda}\sin\theta}, \dots, 
e^{j\frac{2\pi d}{\lambda}(N_B-1)\sin\theta} \right]^T$. The resulting transmit beampattern at angle $\theta$ is defined as $\mathbb{E} \{|\boldsymbol{\nu}^H (\theta) \mathbf{s}|^2\}$. 
Let $P_d(\theta)$ denote the desired beampattern, which is typically designed to concentrate energy toward sensing targets while suppressing sidelobes in other directions. To ensure sensing performance, we discretize the angular domain into a finite set $\Theta$.
The transmit covariance matrix is given by $\mathbb{E}\{\mathbf{s}\mathbf{s}^H\}
= \mathbf{W}_0 \mathbf{W}_0^H + \sum_{k \in \mathcal{K}} \mathbf{w}_k \mathbf{w}_k^H$. The sensing quality is quantified by the beampattern matching error defined as
\begin{align}
E(\mathbf{W}, \delta) = \frac{1}{|\mathrm{\Theta} |} \sum_{\theta \in \mathrm{\Theta}} 
\left| \boldsymbol{\nu}^H(\theta) \left( \mathbf{W}_0 \mathbf{W}_0^H + \sum_{k \in \mathcal{K}} \mathbf{w}_k \mathbf{w}_k^H \right) \boldsymbol{\nu}(\theta) - \delta P_d(\theta) \right|^2,
\end{align}
where $\delta > 0$ is a scaling factor.
 
\subsection{Problem Formulation} 
The objective is to design the BE transmit beamformer to maximize the minimum secrecy rate among all CUs, thereby ensuring fairness in secure communication. Meanwhile, the design must satisfy practical constraints, including the beampattern matching error, and the BS transmit power budget.
The max--min secrecy rate optimization problem is formulated as
\begin{subequations}
\label{eq:P0}
\begin{align}
\max_{\mathbf{W}, \delta} \quad & \min_{k \in \mathcal{K}} R^{\mathrm{sec}}_k(\mathbf{W}) \label{eq:P0-R} \\
\text{s.t.} \quad 
     & E(\mathbf{W}, \delta) \leq E^{\max}, \label{eq:P0-E} \\
     & P_{{b}}(\mathbf{W}) = P_b^{\max}. \label{eq:P0-P}
\end{align}
\end{subequations}
Notice that problem \eqref{eq:P0} is highly non-convex due to the non-concave max--min secrecy rate objective \eqref{eq:P0-R}, which involves coupled signal and interference terms, as well as the non-convex sensing constraint \eqref{eq:P0-E}. These challenges make the problem difficult to solve using conventional convex optimization techniques.

\section{Proposed Global Beamforming Algorithm}

\subsection{Problem Reformulation}

First, we introduce auxiliary variables $\boldsymbol{\alpha} \triangleq [\alpha_1,\dots,\alpha_K]$, $\boldsymbol{\beta} \triangleq [\beta_1,\dots,\beta_K]$,  and $\gamma$.  Then, problem \eqref{eq:P0} can be equivalently reformulated as
\begin{subequations}
\begin{align}
\max_{\mathbf{W}, \boldsymbol{\alpha}, \boldsymbol{\beta}, \gamma, \delta} \quad 
& \gamma \\
\text{s.t.} \quad 
& R_k(\mathbf{W}) \ge \alpha_k, \quad \forall k \in \mathcal{K}, \\
& R_{j,k}(\mathbf{W}) \le \beta_k, \quad  \forall j \in \mathcal{J}, \forall k \in \mathcal{K},\\
& \alpha_k - \beta_k \ge \gamma, \quad \forall k \in \mathcal{K}, \\
& \eqref{eq:P0-E}, \eqref{eq:P0-P}.
\end{align}
\end{subequations}
Define the covariance matrices
\begin{align}
\mathbf{V}_0 &= \mathbf{W}_0 \mathbf{W}_0^H, \quad
\mathbf{V}_k = \mathbf{w}_k \mathbf{w}_k^H, \quad \mathbf{V} = [\mathbf{V}_k]_{k \in \mathcal{K}}, \\
\mathbf{R} &= \sum_{k=0}^{K} \mathbf{V}_k, \quad
\mathbf{G}_{d, k} = \frac{\mathbf{h}_{d,k} \mathbf{h}_{d,k} ^H}{\sigma_{d,k}^2}, \quad \forall d \in \{c,s\}.
\end{align}
Then, the problem can be reformulated as
\begin{subequations}
\label{eq:P1}
\begin{align}
\max_{\mathbf{V}, \mathbf{R}, \boldsymbol{\alpha}, \boldsymbol{\beta}, \gamma, \delta} \quad 
& \gamma \\
\text{s.t.} \quad \quad \
& \ln \! \left(1 + \frac{\mathrm{tr}(\mathbf{G}_{c,k} \mathbf{V}_k)}
{\mathrm{tr}(\mathbf{G}_{c, k} (\mathbf{R} -  \mathbf{V}_k)) + 1} \right) \ge \alpha_k , \quad \forall k \in \mathcal{K}, \label{eq:P1-alpha} \\
& \ln \! \left(1 + \frac{\mathrm{tr}(\mathbf{G}_{s, j} \mathbf{V}_k)}
{\mathrm{tr}(\mathbf{G}_{s, j} (\mathbf{R} - \mathbf{V}_k))  + 1} \right) \le \beta_k, \quad \forall j \in \mathcal{J}, \forall k \in \mathcal{K}, \label{eq:P1-beta} \\
& \alpha_k - \beta_k \ge \gamma, \quad \forall k \in \mathcal{K}, \label{eq:P1-gamma} \\
& \sum_{\theta \in \mathrm{\Theta} } \left| \boldsymbol{\nu}^H (\theta) \mathbf{R} \boldsymbol{\nu}(\theta) - \delta P_d(\theta) \right|^2 \le | \mathrm{\Theta}  | E^{\max}, \label{eq:P1-E} \\
& \mathrm{tr} (\mathbf{R}) \le P_b^{\max}, \label{eq:P1-P} \\
& \mathbf{R} \succeq \sum_{k \in \mathcal{K}} \mathbf{V}_k, \label{eq:P1-R} \\
& \mathbf{V}_k \succeq 0, \quad \forall k \in \mathcal{K}, \label{eq:P1-V} \\
& \mathrm{rank}(\mathbf{V}_k) = 1, \quad \forall k \in \mathcal{K}. \label{eq:P1-rank}
\end{align}
\end{subequations}
Problem \eqref{eq:P1} remains challenging due to the non-convex rank-one constraints. By relaxing these constraints, we obtain the following problem:
\begin{align}
\label{eq:P2}
\max_{\mathbf{V}, \mathbf{R}, \boldsymbol{\alpha}, \boldsymbol{\beta}, \gamma, \delta} \quad \gamma 
\quad \text{s.t. } \eqref{eq:P1-alpha}-\eqref{eq:P1-V}.
\end{align}

\begin{theorem}
Let $\{\mathbf{V}^\star, \mathbf{R}^\star, \boldsymbol{\alpha}^\star, \boldsymbol{\beta}^\star, \gamma^\star, \delta^\star\}$ be a globally optimal solution of problem \eqref{eq:P2}. Then, there exists a globally optimal solution $\{\tilde{\mathbf{V}}, \mathbf{R}^\star, \boldsymbol{\alpha}^\star, \boldsymbol{\beta}^\star, \gamma^\star, \delta^\star\}$ such that $\mathrm{rank}(\tilde{\mathbf{V}}_k) = 1$ for all $k \in \mathcal{K}$, where
\begin{align}
 \tilde{\mathbf{V}}_k &= \frac{\mathbf{V}_k^* \mathbf{G}_{c,k} \mathbf{V}_k^*}
{\mathrm{tr} (\mathbf{G}_{c,k} \mathbf{V}_k^*)}, \quad \forall k \in \mathcal{K},  \\
 \tilde{\mathbf{V}} &= [\tilde{\mathbf{V}}_k ]_{  k \in \mathcal{K}}.
\end{align}
\end{theorem}

\begin{proof}
Please refer to \cite{liu2020joint, ren2022optimal}.
\end{proof}

However, problem \eqref{eq:P2} is still non-convex due to constraints \eqref{eq:P1-alpha} and \eqref{eq:P1-beta}. By introducing auxiliary variables $\mathbf{a} \in \mathbb{R}^K$ and $\mathbf{b} \in \mathbb{R}^{J}$, problem \eqref{eq:P2} can be reformulated as
\begin{subequations}
\label{eq:P3}
\begin{align}
\max_{\mathbf{V}, \mathbf{R}, \mathbf{a}, \mathbf{b}, \boldsymbol{\alpha}, \boldsymbol{\beta}, \gamma, \delta} \quad 
& \gamma \\
\text{s.t.} \quad\quad\quad
& \ln\!\left(1 + \mathrm{tr}(\mathbf{G}_{c,k} (\mathbf{R} - \mathbf{V}_k)) \right) \le a_k , \quad \forall k \in \mathcal{K}, \label{eq:P3-a} \\
& \ln\!\left(1 + \mathrm{tr}(\mathbf{G}_{s,j} \mathbf{R}) \right) \le b_{j}, \quad \forall j \in \mathcal{J}, \label{eq:P3-b} \\
& \ln\!\left(1 + \mathrm{tr}(\mathbf{G}_{c,k} \mathbf{R}) \right) - a_k \ge \alpha_k , \quad \forall k \in \mathcal{K}, \label{eq:P3-alpha} \\
& b_{j} - \ln\!\left(1 + \mathrm{tr}(\mathbf{G}_{s,j} (\mathbf{R} - \mathbf{V}_k)) \right) \le \beta_k, \quad \forall j \in \mathcal{J}, \forall k \in \mathcal{K}, \label{eq:P3-beta} \\
& \eqref{eq:P1-gamma} - \eqref{eq:P1-V}.
\end{align}
\end{subequations}
It can be observed that problem \eqref{eq:P3} remains non-convex due to the non-convex constraints involving concave logarithmic functions. In the next step, we develop a convex relaxation of \eqref{eq:P3} based on convex hull function.

\subsection{Convex Relaxation}

In problem \eqref{eq:P3}, constraints \eqref{eq:P3-a} and \eqref{eq:P3-b} are non-convex due to the presence of concave logarithmic functions in inequality constraints. To develop a BB algorithm, we construct convex relaxations for these constraints \cite{tuy1998convex}. Specifically, these constraints can be rewritten in the form of the non-convex set $\mathcal{D} = \{(x, y) \in \mathbb{R}_+^2 \mid y \le e^x \}$.
Assuming $x \in [l,u]$, define $\mathcal{D}_{[l,u]} = \{(x,y) \mid y \le e^x, \ x \in [l,u] \}$. The convex envelope of this set, denoted by $\mathrm{Conv}(\mathcal{D}_{[l,u]})$, provides the tightest convex relaxation over $[l,u]$. Since $e^x$ is convex, we have
\begin{align}
e^x \le f_{[l,u]}(x) \triangleq \frac{e^u - e^l}{u - l}(x - l) + e^l, \quad \forall x \in [l,u].
\end{align}
Thus, the convex envelope is given by $\mathrm{Conv}(\mathcal{D}_{[l,u]}) = \{(x,y) \in \mathbb{R}_+^2 \mid y \le f_{[l,u]}(x), \ x \in [l,u]\}$ (see Fig.~\ref{fig:con-rel}).

\begin{figure}[htbp]
    \centering
    \includegraphics[width=0.9\textwidth]{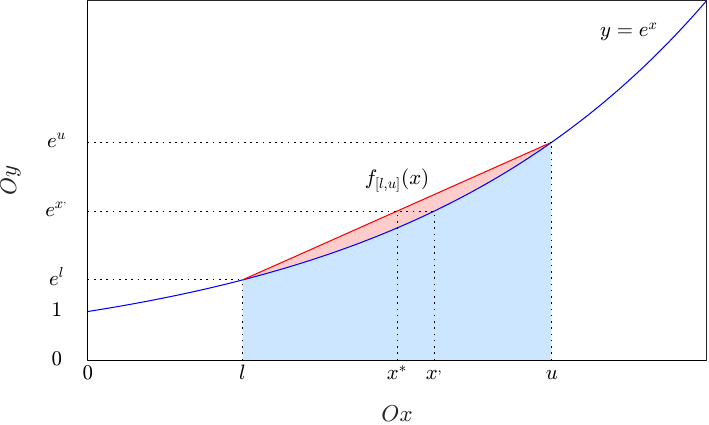}
    \caption{Convex envelope of the non-convex set $\mathcal{D}_{[l,u]}$.}
    \label{fig:con-rel}
\end{figure}
Using this result, we obtain the following convex relaxation of problem \eqref{eq:P3}:
\begin{subequations}
\label{eq:P4}
\begin{align}
\max_{\mathbf{V}, \mathbf{R}, \mathbf{a}, \mathbf{b}, \boldsymbol{\alpha}, \boldsymbol{\beta}, \gamma, \delta} \quad 
& \gamma \\
\text{s.t.} \quad\quad\quad
& 1 + \mathrm{tr}(\mathbf{G}_{c,k} (\mathbf{R} - \mathbf{V}_k)) 
\le f_{[a_k^l, a_k^u]} (a_k), \quad \forall k \in \mathcal{K}, \\
& 1 + \mathrm{tr}(\mathbf{G}_{s,j} \mathbf{R}) 
\le f_{[b_j^l, b_j^u]} (b_j), \quad \forall j \in \mathcal{J}, \\
& \mathbf{a}^{l} \le \mathbf{a} \le \mathbf{a}^{u}, \\
& \mathbf{b}^{l} \le \mathbf{b} \le \mathbf{b}^{u}, \\
& \eqref{eq:P3-alpha}, \eqref{eq:P3-beta}, \eqref{eq:P1-gamma} - \eqref{eq:P1-V}.
\end{align}
\end{subequations}
Problem \eqref{eq:P4} is a convex optimization problem, as the objective function is linear and all constraints are convex. Therefore, it can be efficiently solved using standard convex optimization solvers such as CVXPY \cite{diamond2016cvxpy}. 

\subsection{Proposed BB Algorithm}
To globally solve problem \eqref{eq:P3}, we adopt a BB framework based on the convex relaxation \eqref{eq:P4}. Let $\mathcal{B} = [\mathbf{a}^{(l)}, \mathbf{a}^{(u)}] \times [\mathbf{b}^{(l)}, \mathbf{b}^{(u)}]$ denote a box (hyper-rectangle) in the space of auxiliary variables $\mathbf{a}, \mathbf{b}$. The BB algorithm proceeds by recursively partitioning $\mathcal{B}$ and computing upper and lower bounds over each subregion.

\paragraph{Initialization} We initialize the box $\mathcal{Q}_0 = \prod_{k \in \mathcal{K}} [0, a_k^{max} ] \times \prod_{j \in \mathcal{J}} [0, b_j^{max}] $, where
\begin{align}
\label{eq:B0}
a_{k}^{max} = \ln\!\left({1 + \mathrm{tr}(\mathbf{G}_{c,k}) P_b^{max} } \right), \quad k \in \mathcal{K}, \\
b_{j}^{max} = \ln\!\left({1 + \mathrm{tr}(\mathbf{G}_{s,j} ) P_b^{max} } \right), \quad j \in \mathcal{J}.
\end{align}

\paragraph{Bounding}
For a given box $\mathcal{B}$, the upper bound is obtained by solving the convex relaxation \eqref{eq:P4}, denoted by $\mathrm{UB}(\mathcal{B})$. Let $\{\mathbf{V}^\star, \mathbf{R}^\star, \mathbf{a}^\star, \mathbf{b}^\star, \boldsymbol{\alpha}^\star, \boldsymbol{\beta}^\star, \gamma^\star, \delta^\star \}$ be the optimal solution of \eqref{eq:P4}. A feasible solution $\mathrm{S} (\mathcal{B}) = \{ \mathbf{V}^\star, \mathbf{R}^\star, \mathbf{a}', \mathbf{b}', \boldsymbol{\alpha}', \boldsymbol{\beta}', \gamma', \delta^\star \}$ to the original problem \eqref{eq:P3} can be constructed as
\begin{align}
\label{eq:LB}
a'_k &= \ln (f_{[a^l_k,a^u_k]} (a_k^*)), \quad \forall k \in \mathcal{K}, \\    
b'_j &= \ln (f_{[b^l_j,b^u_j]} (b_j^*)), \quad \forall j \in \mathcal{J}, \\
\alpha_k' &= \alpha_k^\star - (a'_k - a_k^*), \quad k \in \mathcal{K}, \\
\beta_k' &= \beta_k^* + \max_{j \in \mathcal{J}} (b_j' - b_j^*)  , \quad k \in \mathcal{K}, \\
\gamma' &= \gamma^* - \max_{k \in \mathcal{K}} (a_k' - a_k^\star) - \max_{j \in \mathcal{J}} (b_j' - b_j^\star).
\end{align}
This construction ensures feasibility since the convex envelope provides a valid upper bound of the exponential function over the box. The resulting objective value yields a lower bound $\mathrm{LB}(\mathcal{B}) = \gamma'$.

\paragraph{Termination}
Let $U^{(\iota)}$ and $L^{(\iota)}$ denote the best upper and lower bounds at iteration $\iota$, respectively. The algorithm terminates when
\begin{align}
\label{eq:Tem}
U^{(\iota)} - L^{(\iota)} \le \epsilon,
\end{align}
where $\epsilon > 0$ is the given error tolerance.

\paragraph{Branching}
If the stopping criterion in \eqref{eq:Tem} is not satisfied, the algorithm proceeds by branching the current box. At each iteration, the box with the largest upper bound is selected for partitioning. 
Let $ \mathbf{z}^u = [\mathbf{a}^{*T}, \mathbf{b}^{*T}]^T, \mathbf{z}^l = [\mathbf{a}'^T, \mathbf{b}'^T]^T$ denote the solutions obtained from the upper bound problem and the constructed feasible solution  associated with a box $\mathcal{B}$, respectively. The branching dimension is chosen as
\begin{align}
d^\star = \arg\max_{d} \left\{ z_d^{u} - z_d^{l} \right\},
\end{align}
i.e., the dimension with the largest interval length. Let $m_{d^\star} = \left( z_{d^\star}^{l} + z_{d^\star}^{u} \right)/2$ denote the midpoint. The box $\mathcal{B}^{(\iota)}$ is then partitioned into two sub-boxes:
\begin{subequations}
\label{eq:Branch}
\begin{align}
\mathcal{B}_l^{(\iota)} &= \left\{ \mathbf{z} \in \mathcal{B}^{(\iota)} \mid z_{d^\star} \le m_{d^\star} \right\}, \\
\mathcal{B}_u^{(\iota)} &= \left\{ \mathbf{z} \in \mathcal{B}^{(\iota)} \mid z_{d^\star} \ge m_{d^\star} \right\}.
\end{align}
\end{subequations}

\subsection{Global Convergence and Worst-Case Iteration Complexity}

In this subsection, we present theoretical guarantees for the proposed BB algorithm. We first define the notion of an $\epsilon$-optimal solution to problem \eqref{eq:P3}.

\begin{definition}
Given any $\epsilon > 0$, a feasible solution 
$\{ \mathbf{V}, \mathbf{R}, \mathbf{a}, \mathbf{b}, \boldsymbol{\alpha}, \boldsymbol{\beta}, \gamma, \delta \}$ 
is called an $\epsilon$-optimal solution of problem \eqref{eq:P3} if it satisfies
\begin{align}
\gamma^* - \gamma \le \epsilon,
\end{align}
where $\gamma^*$ denotes the global optimal value of problem \eqref{eq:P3}.
\end{definition}
\begin{theorem}
If the BB algorithm terminates, the returned solution $S$ by the algorithm is an $\epsilon$-optimal solution of problem \eqref{eq:P3}.
\end{theorem}
\begin{proof}
At iteration $\iota$-th of the BB algorithm, the feasible region $\mathcal{Q}^{(0)}$ is partitioned into a list $\mathcal{L}^{(\iota)} $ of sub-boxes and the global solution of the given problem instance must lie in one of them, we denote the corresponding subset as $\mathcal{B^*}$ ($\mathcal{B}^* \in \mathcal{L}^{(\iota)}$). By construction of the BB framework, the upper bound satisfies $ U^{(\iota)} = \max_{\mathcal{B} \in \mathcal{L^{(\iota)}}} UB (\mathcal{B}) \ge UB (\mathcal{B}^*) \ge \gamma^*$, since it is obtained from the relaxation of the original problem. Meanwhile, the lower bound satisfies $L^{(\iota)} \le \gamma^*$, as it is obtained from a feasible solution of the original problem. Therefore, we have $L^{(\iota)} \le \gamma^* \le U^{(\iota)}$. The algorithm terminates when $U^{(\iota)} - L^{(\iota)} \le \epsilon$. Combining the above inequalities yields $0 \le \gamma^* - L^{(\iota)} \le U^{(\iota)} - L^{(\iota)} \le \epsilon$, which implies that the obtained solution $\mathcal{S}^{(\iota)} $ is $\epsilon$-optimal. This completes the proof.
\end{proof} 

\begin{algorithm}[htb]
\caption{Iterative BB Algorithm solving Problem \eqref{eq:P3}.}
\begin{algorithmic}[1]
\label{Alg:brb}
\STATE \textbf{Initialization}: 
\STATE Set a tolerance $\epsilon > 0$; 
\STATE Initialize box $\mathcal{Q}^{(0)}$ as \eqref{eq:B0}, $\mathcal{P}^{(1)} = \{\mathcal{Q}^{(0)}\}, \mathcal{L} = \varnothing, L^{(0)} = 0$. 
\FOR {$\iota = 1, 2, ...$}
\FOR {each box $\mathcal{B}$ in $\mathcal{P}^{(\iota)} $}
\STATE {\it // Bounding}
\STATE Solve the convex problem \eqref{eq:P4} over $\mathcal{B}$ to update upper bound $\mathrm{UB}(\mathcal{B})$;
\STATE Follow \eqref{eq:LB} to get solution $\mathrm{S}(\mathcal{B})$ and lower bound $\mathrm{LB}(\mathcal{B})$;
\STATE {\it // Pruning}
\IF{$\mathrm{UB}(\mathcal{B}) > L^{(\iota)} $}
\STATE Add box $\mathcal{B}$ into the list $\mathcal{L}$.
\ENDIF 
\STATE {\it // Update lower bound}
\IF{$\mathrm{LB}(\mathcal{B}) > L^{(\iota)} $}
\STATE Update $L^{(\iota)} \leftarrow \mathrm{LB}(\mathcal{B})$ and $\mathcal{S}^{(\iota)} \leftarrow \mathrm{S}(\mathcal{B})$;
\ENDIF
\ENDFOR
\STATE Choose the maximum upper bound $U^{(\iota)} = \max_{\mathcal{B} \in \mathcal{L}} UB(\mathcal{B}) $
\STATE {\it // Termination} 
\IF {$ U^{(\iota)} - L^{(\iota)} \leq \epsilon$}
\STATE Terminate the algorithm
\ENDIF
\STATE {\it // Branching}
\STATE Choose the box $\mathcal{Q}^{(\iota)}$ that has the upper bound $U^{(\iota)}$ from $\mathcal{L}^{(\iota)}$;
\STATE Branch $\mathcal{Q}^{(\iota)}$ into 2 sets $\mathcal{B}_l^{(\iota)}, \mathcal{B}_u^{(\iota)}$ follow \eqref{eq:Branch}; 
\STATE Set $\mathcal{P}^{(\iota+1)} = \{\mathcal{B}_l^{(\iota)}, \mathcal{B}_u^{(\iota)} \}, L^{(\iota+1)} =  L^{(\iota)} \setminus \{ \mathcal{Q}^{(\iota)}\}$; 
\ENDFOR
\STATE \textbf{Output}:
Return the optimized BS beamforming matrix. 
\end{algorithmic}
\end{algorithm}

\begin{lemma}
For any given $\epsilon > 0$, the BB algorithm will return an $\epsilon$-optimal solution if
\begin{align}
\max \{ \max_{k \in \mathcal{K}} (a_{k}^u - a_{k}^l), \max_{j \in \mathcal{J}} (b_{j}^u - b_{j}^l) \} \le \omega = W_0 (\rho) - W_{-1} (\rho)
\end{align}
where $\rho = -1/\mathrm{e}^{1+\epsilon/2} \ge -1/e$ and $W_0(.), W_{-1} (.) $ are Lambert function. 
\end{lemma} 
\begin{proof}
First, we consider $x \in [x_0, x_0 + \omega]$. By applying the first-order optimality condition, we obtain an upper bound on $\ln (f_{[x_0, x_0 + \omega]} (x)) - x$ as
\begin{align}
\ln (f_{[x_0, x_0 + \omega]} (x)) - x 
&\le \bar{f} (\omega) \triangleq  \ln \!\left(\frac{\mathrm{e}^\omega-1}{\omega}\right) + \frac{\omega}{\mathrm{e}^\omega-1} - 1, \quad \forall x \in [x_0, x_0 + \omega ].
\end{align}
Since $W_0 (\rho) \mathrm{e}^{W_0 (\rho)} = W_{-1} (\rho) \mathrm{e}^{W_{-1} (\rho)} = \rho $, we have
\begin{align}
\frac{\mathrm{e}^\omega-1}{\omega} &= \frac{\mathrm{e}^{W_0 (\rho) - W_{-1} (\rho)} -1}{W_0 (\rho) - W_{-1} (\rho)} = \frac{\frac{{W_{-1} (\rho)}}{{W_{0} (\rho)}} - 1}{W_0 (\rho) - W_{-1} (\rho)}  = -\frac{1}{W_0(\rho)}
\end{align}
Hence $ \bar{f} (\omega) = \ln \left(- {1}/{W_0(\rho)} \right) - {W_0(\rho)} -1 = -\ln (- {W_0(\rho)} \mathrm{e}^{{W_0(\rho)}} ) -1 = -\ln (- \rho) -1 = \epsilon/2 $. Let $\gamma^*, \gamma'$ be the solution chosen to maximize upper bound 
\begin{align}
\gamma^* - \gamma' &= \max_{k \in \mathcal{K}} (a_k' - a_k^\star) + \max_{j \in \mathcal{J}} (b_j' - b_j^\star) \\
& \le \max_{k \in \mathcal{K}} \bar{f} ( a_k^u - a_k^l) + \max_{j \in \mathcal{J}} \bar{f} ( b_j^u - b_j^l) \\
& \le \epsilon/2 + \epsilon/2 = \epsilon
\end{align}
Therefore, $U^{(\iota)} = \gamma^*, L^{(\iota)} \ge \gamma',$ the algorithm terminates when $U^{(\iota)} - L^{(\iota)} \le \gamma^* -\gamma' \le \epsilon$. This completes the proof.
\end{proof}

\begin{theorem}
For any given $\epsilon > 0$, the BB algorithm will return an $\epsilon$-optimal solution of the given instance within at most $I_{max}$
iterations, where
\begin{align}
    I_{max} =  \left\lceil \prod_{k \in \mathcal{K}} \frac{a_{max} }{\omega} \times \prod_{j\in\mathcal{J}} \frac{b_{max} }{\omega} \right\rceil + 1
\end{align}
\end{theorem}
\begin{proof}
Suppose that the algorithm does not terminate within $I_{max}$ iterations. This fact, together with Lemma 1, implies that the interval that is chosen to be partitioned at the $\iota$-th iteration must satisfy max edge of every box is greater than $\omega$. Then, after the partition, the width of the two sub-intervals is greater than $\omega/2$. Based on this, we can conclude that, for each subset partitioned from the
original set $\mathcal{Q}^{(0)}$, there holds every edge is greater than $\omega$. Hence, the volume of each subset is not less than $\omega^{K+J}$ and the total volume of all $I_{max}$ subsets is not less than $I_{max}\omega^{K+J} > \prod_{k \in \mathcal{K}} {a_{max} } \times \prod_{j\in\mathcal{J}} {b_{max} }$.
Obviously, the volume of $\mathcal{Q}^{(0)}$ is $\mathrm{vol}(\mathcal{Q}^{(0)})  =\prod_{k \in \mathcal{K}} {a_{max} } \times \prod_{j\in\mathcal{J}} {b_{max} }$, which further implies that the total volume of all subsets is greater than the one of the original set. This is a contradiction. Hence, the algorithm will terminate within at most $I_{max}$ iterations.
\end{proof}

\begin{corollary}
For any given instance of problem, let $\gamma^*$ be its optimal value and let $\{U^{(\iota)}\}$ and $\{L^{(\iota)}\}$ be the iterates generated by the BB algorithm with $\epsilon = 0$. Then, we have $\{U^{(\iota)}\} \to \gamma^*, \{L^{(\iota)}\} \to \gamma^*$.
\end{corollary}
\begin{proof}
Theorem 3 shows that, for any given $\epsilon > 0$, there exists an integer $I_{max}$ such that $U^{(\iota)} - L^{(\iota)} < \epsilon $ for all $\iota \ge I_{max}$. This statement is equivalent to $\{U^{(\iota)} - L^{(\iota)}\} \to 0$. Together with Theorem 2 ($L^{(\iota)} \le \gamma^* \le U^{(\iota)}$), this completes the proof.
\end{proof}
\section{Low-Complexity Algorithm}
\subsection{{Proposed SCA-Based Algorithm}}
The proposed BB framework can achieve the globally optimal solution of Problem \eqref{eq:P3}. Nevertheless, its worst-case computational complexity grows exponentially with the number of branching variables, which may become prohibitive for large-scale ISAC systems. To provide a computationally efficient alternative, we develop a low-complexity algorithm based on SCA.

Since $e^x$ is a convex function, its first-order Taylor expansion at a given point $\bar{x}$ provides a lower bound $e^x \ge e^{\bar{x}} + e^{\bar{x}}(x-\bar{x}), \quad \forall x$. Accordingly, at iteration $\iota$, we approximate the non-convex constraints \eqref{eq:P3-alpha} and \eqref{eq:P3-beta} by the linear constraints
\begin{align}
& 1 + \mathrm{tr}(\mathbf{G}_{c,k} (\mathbf{R} - \mathbf{V}_k)) \le e^{\bar{a}_k} (a_k - \bar{a}_k + 1 ), \quad \forall k \in \mathcal{K}, \label{eq:P5-a} \\
& 1 + \mathrm{tr}(\mathbf{G}_{s,j} \mathbf{R}) \le e^{\bar{b}_j} (b_j - \bar{b}_j + 1 ) , \quad \forall j \in \mathcal{J}, \label{eq:P5-b}
\end{align}
where $\bar a_k$ and $\bar b_j$ denote the values obtained from the previous iteration. As a result, the problem \eqref{eq:P3} can be approximated by the following convex program
\begin{subequations}
\label{eq:P5}
\begin{align}
\max_{\mathbf{V}, \mathbf{R}, \mathbf{a}, \mathbf{b}, \boldsymbol{\alpha}, \boldsymbol{\beta}, \gamma, \delta} \quad 
& \gamma \\
\text{s.t.} \quad\quad\quad
& \eqref{eq:P5-a}, \eqref{eq:P5-b} ,  \eqref{eq:P3-alpha}, \eqref{eq:P3-beta}, \eqref{eq:P1-gamma} - \eqref{eq:P1-V}.
\end{align}
\end{subequations}

Problem \eqref{eq:P5} is convex and can be efficiently solved using standard convex program
solvers such as CVXPY \cite{diamond2016cvxpy}. By successively updating the linearization points ${\bar a_k,\bar b_j}$, a sequence of feasible solutions is generated until convergence. The complete procedure is summarized in Algorithm \ref{Alg:sca}.

\begin{algorithm}[htb]
\caption{Iterative SCA Algorithm solving Problem \eqref{eq:P3}.}
\begin{algorithmic}[1]
\label{Alg:sca}
\STATE \textbf{Initialization}: 
\STATE Set a tolerance $\epsilon > 0$; 
\STATE Initialize $[\mathbf{a}^{(0)}; \mathbf{b}^{(0)}] \in \mathcal{Q}^{(0)}$ as \eqref{eq:B0}, $\gamma^{(0)} = -\infty $. 
\FOR {$\iota = 1, 2, ...$}
\STATE Update $\bar{\mathbf{a}} \leftarrow \mathbf{a}^{(\iota-1)} $ and $\bar{\mathbf{b}} \leftarrow \mathbf{b}^{(\iota-1)} $;
\STATE Solve the convex problem \eqref{eq:P5} to update $\mathbf{V}^{(\iota)} , \mathbf{R}^{(\iota)}, \mathbf{a}^{(\iota)}, \mathbf{b}^{(\iota)}, \gamma^{(\iota)}$;
\IF{$ \gamma^{(\iota)} -  \gamma^{(\iota-1)} < \epsilon $}
\STATE Terminate the algorithm;
\ENDIF
\ENDFOR
\STATE \textbf{Output}:
Return the optimized BS beamforming matrix $\mathbf{V}^{(\iota)} , \mathbf{R}^{(\iota)}$. 
\end{algorithmic}
\end{algorithm}

\subsection{Convergence and Complexity}
The proposed algorithm belongs to the class of SCA, where at each iteration, non-convex constraints are replaced by their first-order convex approximations to ensure a fit at the current operating point. Thanks to this mathematical characteristic, the solution obtained at the current iteration is always a feasible point for the approximation problem at the next iteration. Consequently, the value of the objective function increases steadily and non-decreasingly throughout the iteration and is bounded above by the power constraint; the algorithm is guaranteed to converge to a stopping point of the non-convex problem \eqref{eq:P3}. 

At each iteration, the SCA algorithm offers a significant advantage by requiring only the solution of a fully convex SDP. By using standard interior-point solvers, the computational complexity for each iteration increases only with polynomial time for the dimensions of the problem, such as the number of base station antennas and the number of CUs. This contrasts sharply with the BB algorithm framework, whose computational complexity increases exponentially based on the number of branching variables in the worst-case scenario. Therefore, the overall complexity of the SCA algorithm is significantly lower than that of the BB method, making it an extremely practical solution for deployment in large-scale ISAC systems while maintaining near-optimal security performance.

\section{Numerical Results}
\subsection{Simulation scenario}
The simulation scenario considers a system where the BS is located at (0, 0) m, equipped with a ULA of $N_B=16$ antennas. 
The system serves $K = 3$ CUs, specifically CUs moving within a circle centered at (100, 0) with a radius of 50 meters. For sensing purposes, there are $J=3$ sensing targets at angles of $-60^{\circ}, 0^{\circ}, 60^{\circ}$. For the simulated channel model, Rician fading is considered with a Rician factor of 3 dB. The path loss exponent from the BS to CUs and SUs is 2.2. The noise power at all devices is set to -90 dBm. 

Regarding the system parameters, the maximum transmit power at the BS is $P_B^{max} = 20$ dBm. For sensing functionality, the number of samples is $M=181$, the beamwidth is $\Delta \phi =10 ^\circ$, and the maximum beam pattern error is $E^{max} = -20$ dB. The algorithm is implemented with a minimum error threshold of $10^{-2}$, across 100 different channel realizations. 

\subsection{Convergence Behavior} 

Fig.~\ref{fig:conv} compares the convergence performance of the proposed BB algorithm and the low-complexity SCA method. The gap between the UB and LB of BB algorithm decreases rapidly and becomes nearly zero after approximately 110 iterations, verifying convergence to the globally optimal solution. In contrast, the SCA algorithm converges significantly faster while achieving a solution that is very close to the global optimum obtained by the proposed BB algorithm. Furthermore, repeated simulations with different initialization points yield nearly identical objective values, indicating that the SCA method is relatively insensitive to initialization and provides a low-complexity alternative.

\begin{figure}[htb]
    \centering
    \includegraphics[width=0.6\textwidth]{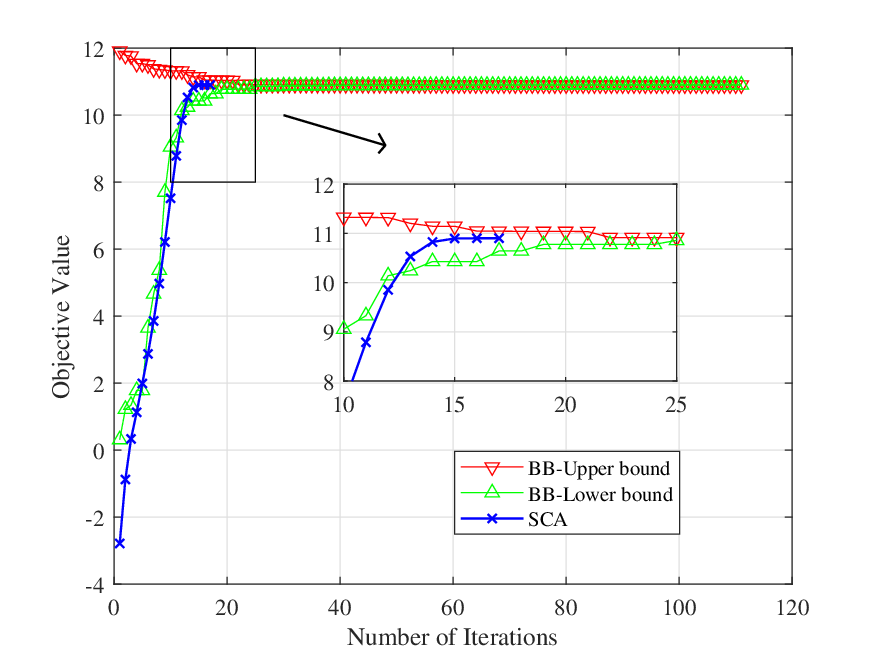}
    \caption{Convergence behavior of the proposed BB and SCA algorithms.}
    \label{fig:conv}
\end{figure}


\begin{figure}[htbp]
    \centering
    \includegraphics[width=0.6\textwidth]{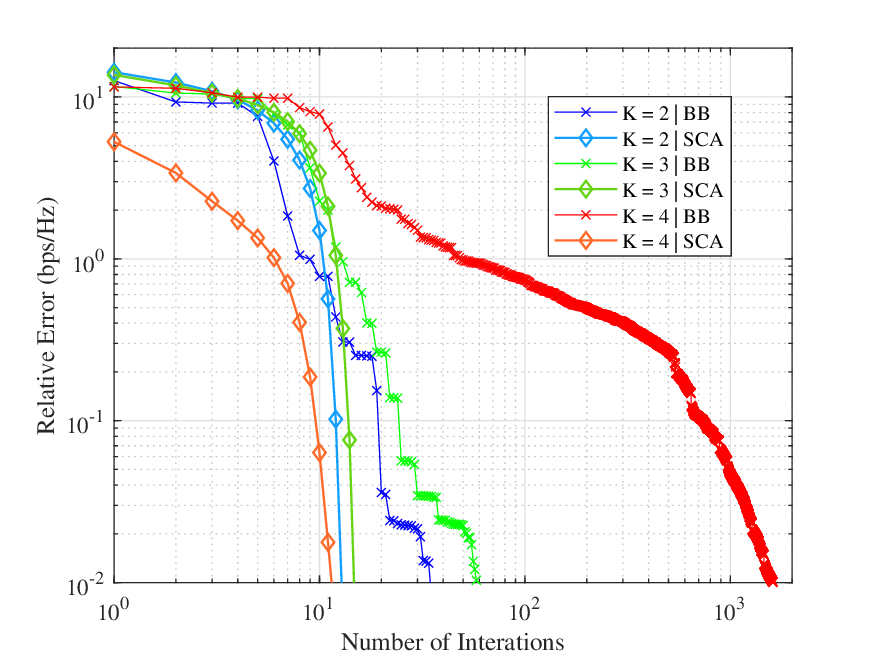}
    \caption{Relative error of the proposed BB algorithm and SCA algorithm.}
    \label{fig:error}
\end{figure}

Fig.~\ref{fig:error} shows the relative error of the proposed BB algorithm and the low-complexity SCA method for different numbers of Cus $K$, where the relative error is computed with respect to the final upper bound obtained by BB. For all considered cases, the relative error decreases monotonically, confirming the convergence of both algorithms. When $K=2$ and $K=3$, both methods converge rapidly and achieve a relative error below $10^{-2}$ within fewer than 100 iterations. Moreover, the convergence curves of SCA closely match those of BB, indicating that the locally optimal solutions found by SCA are very close to the global optimum. As the number of CUs increases to $K=4$, the convergence behavior of BB becomes significantly slower. More than $10^3$ iterations are required to reduce the relative error below $10^{-2} $, reflecting the increased complexity of the global search process as the problem dimension grows. In contrast, SCA still converges within approximately 10 iterations and reaches a solution with a small optimality gap. This observation suggests that, although the worst-case complexity of global optimization increases rapidly with the number of users, the low-complexity SCA algorithm remains an effective and computationally attractive alternative.

\subsection{Analyze the impact of the parameters}
Figure \ref{fig:scan_E} presents a comprehensive and detailed picture of the stringent resource trade-offs in a multi-user ISAC network. The diagram clearly illustrates the contradiction between the two functions: when the beampattern error requirements are tightened, the minimum secrecy rate in all scenarios is drastically reduced to near zero because BS has to devote almost all its energy and spatial degrees of freedom to target scanning; conversely, when the sensing requirements are loosened, the secrecy rate increases sharply but quickly enters a saturation state from -20 dB onwards, indicating that the system is no longer constrained by radar limitations but has reached its maximum transmit power limit or is affected by interference between users.

\begin{figure}[!htb]
    \centering
    \includegraphics[width=0.6\textwidth]{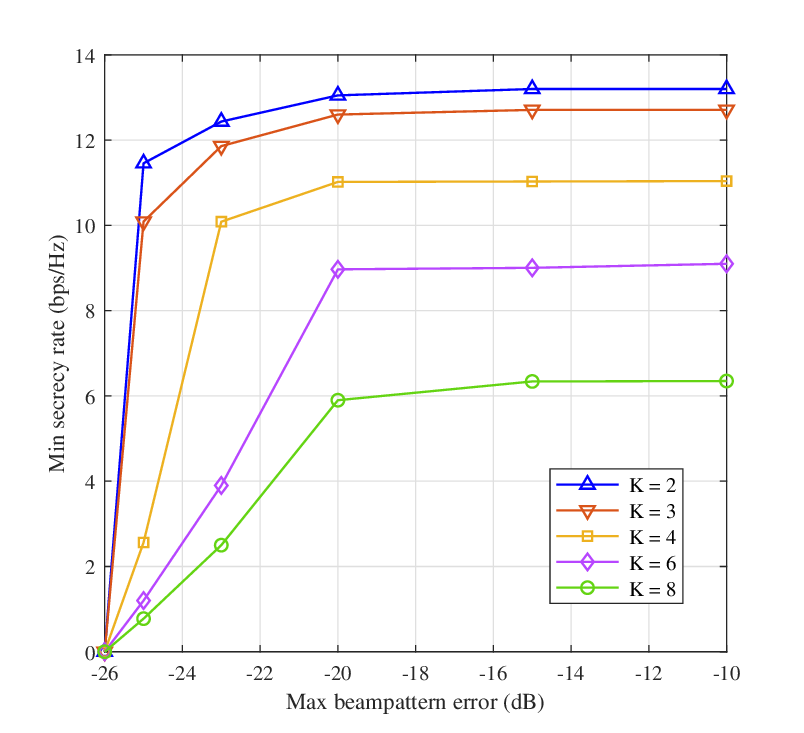}
    \caption{Min secrecy rate versus max beampattern error under different number of CUs.}
    \label{fig:scan_E}
\end{figure}

The graph also clearly shows the physical load limit of the system as the number of CUs increases. With the BS station only equipped with $N_B = 16$ antennas, when having to serve a large number of CUs simultaneously while simultaneously creating beam patterns to track 3 SUs, the system becomes overloaded due to the exhaustion of spatial degrees of freedom, leading to extremely low security performance despite loosening the sensing error. The results from Fig.~\ref{fig:scan_E} are an important database that helps determine the equilibrium operating point (e.g., at a tolerance level of $-23$ dB), thereby ensuring a balance between radar quality and optimizing safe transmission capacity for legitimate users.

Figure \ref{fig:scan_N} thoroughly addresses the hardware bottleneck problem in ISAC networks by evaluating the impact of the number of transmitting station antennas ($N_B$) on the minimum security speed. As the number of antennas increases from 10 to 32, the security rate in all tolerance scenarios gradually increases. The graph shows the ability to overcome the strict trade-off: if at $N_B = 16$, the performance difference between the sensor constraint ($E = 0$ dB) and ($E = -25$ dB) is quite large, then at a scale of 25 to 32 antennas, the curve of $E = -25$ dB almost converges with the other curves. This shows that equipping a large-scale antenna system will allow the ISAC network to simultaneously achieve perfect accuracy with almost no trade-off in user secrecy rate.

\begin{figure}[!htb]
    \centering
    \includegraphics[width=0.6\textwidth]{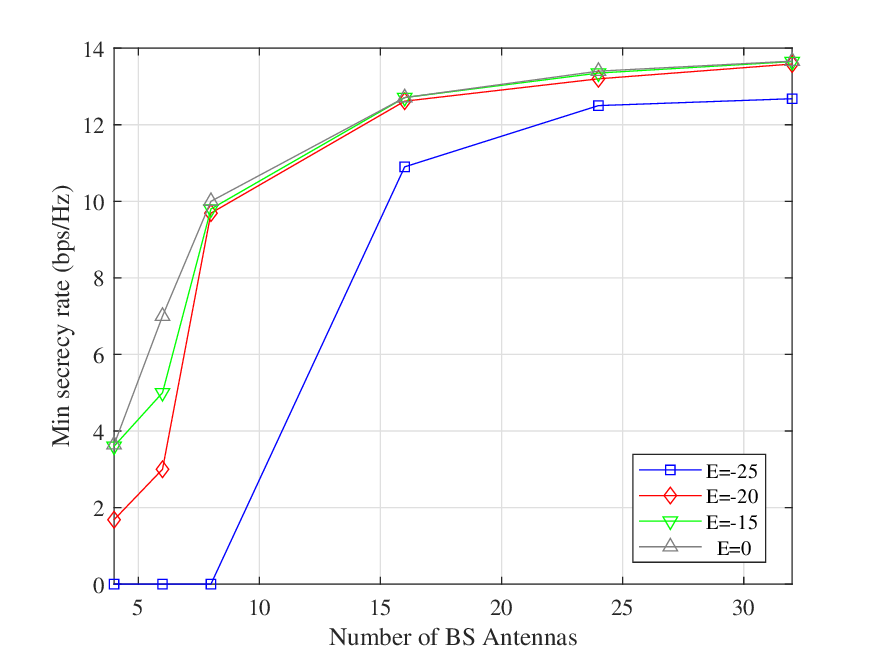}
    \caption{Min secrecy rate versus the number of BS antennas under different max beampattern error.}
    \label{fig:scan_N}
\end{figure}

\section{Conclusions}
In this paper, we investigate a secure ISAC system with unreliable SUs acting as potential eavesdroppers. We construct a max-min secrecy rate optimization problem under constraints on transmit power and sensing quality. To solve this non-convex problem, we develop a global optimization framework based on BB combined with SDR and convex envelope techniques, ensuring convergence to a global optimum within a given tolerance. To address the high computational complexity of the BB framework in large-scale scenarios, we propose a low-complexity algorithm based on the SCA method. The results show that the proposed BB method provides an absolute performance benchmark and the proposed SCA algorithm achieves near-optimal security performance with significantly reduced computational complexity, and becoming a very practical solution for real-world ISAC implementations.

\begin{credits}
\subsubsection{\ackname} 
We would like to thank Viettel High Technology Industries Corporation, Viettel Group for providing essential resources enabling the completion of this work.

\end{credits}
%
%
%
\bibliographystyle{splncs04}
\bibliography{references}

\end{document}